\documentclass[a4paper,11pt]{article}
\usepackage{amsmath}
\usepackage{amsfonts,amssymb}
\usepackage{amsthm}
\usepackage[dvipdfmx]{graphicx}
\usepackage{color}

\newcommand{\bigpare}[1]{\bigl(#1\bigr)}
\newcommand{\biggpare}[1]{\biggl(#1\biggr)}

\newcommand{\bigbrac}[1]{\bigl[#1\bigr]}
\newcommand{\Bigbrac}[1]{\Bigl[#1\Bigr]}
\newcommand{\biggbrac}[1]{\biggl[#1\biggr]}

\newcommand{\bigset}[2]{\bigl\{#1\bigm|#2\bigr\}}

\newcommand{\norm}[1]{\| #1 \|}

\newcommand{\biggabs}[1]{\biggl| #1 \biggr|}

\newcommand{\jap}[1]{\langle #1 \rangle}

\def\a{\alpha}
\def\b{\beta}
\def\c{\gamma}
\def\d{\delta}
\def\e{\varepsilon}

\def\l{\lambda}
\def\m{\mu}
\def\n{\nu}
\def\o{\omega}
\def\s{\sigma}

\def\y{\eta}

\renewcommand{\O}{\Omega}

\def\re{\mathbb{R}}
\def\co{\mathbb{C}}
\def\ze{\mathbb{Z}}

\def\pa{\partial}

\renewcommand{\Re}{\text{{\rm Re}}}
\renewcommand{\Im}{\text{{\rm Im}}}
\newcommand{\supp}{\text{{\rm supp}\;}}

\newcommand{\dist}{\text{\rm dist}}

\newtheorem{thm}{Theorem}
\newtheorem{lem}[thm]{Lemma}

\theoremstyle{definition}

\theoremstyle{remark}

\begin{document}

\title{A Note on the Analyticity of Density of States}
\author{M. Kaminaga\thanks{Department of Electrical Engineering and Information Technology, 
Tohoku Gakuin University, Tagajo, 985-8537, JAPAN. 
Email: \texttt{kaminaga@tjcc.tohoku-gakuin.ac.jp}.}, 
M. Krishna\thanks{Institute of Mathematical Sciences, CIT Campus, Taramani 600 113, Chennai, INDIA. 
Email: \texttt{krishna@imsc.res.in}.}, 
S. Nakamura\thanks{Graduate School of Mathematical Sciences, 
University of Tokyo, 3-8-1 Komaba, Meguroku, Tokyo, 153-8914, JAPAN. 
Email: \texttt{shu@ms.u-tokyo.ac.jp}. Partially supported by JSPS Grant Kiban (A) 21244008.}}
\maketitle

\begin{abstract}
We consider the $d$-dimensional Anderson model, and we prove the density of states is 
locally analytic if the single site potential distribution is locally analytic and the disorder is 
large. We employ the random walk expansion of resolvents and a simple complex function theory trick. 
In particular, we discuss the uniform distribution case, and we obtain a sharper result using more 
precise computations. The method can be also applied to prove the analyticity of 
the correlation functions. 
\end{abstract}

\section{Introduction}
We consider the Anderson tight binding model, i.e., a random Schr\"odinger operator
\[
H^\o =H_0+ V^\o \quad \text{on } \mathcal{H}= \ell^2(\ze^d), 
\]
where $d\geq 1$, $V^\o =\{V^\o(n)\,|\,n\in\ze^d\}$ are i.i.d.\ random variables with the common 
distribution $\m$, and $H_0$ is given by
\[
H_0u(n)= h\sum_{|n-m|=1} u(m)\quad \text{for }u\in\mathcal{H}. 
\]
with a constant $h>0$. 

For a finite box $\Gamma\subset\ze^d$, we denote by $H_{\Gamma}^{\omega}$ the operator 
$H^{\omega}$ restricted to $\ell^2(\Gamma)$ with Dirichlet boundary conditions.
The integrated density of states (IDS for short), ${\cal N}(E)$, is defined by
\begin{displaymath}
{\cal N}(E)=\lim_{\Gamma\to\ze^d}\frac{1}{\#\Gamma}\#\{\mbox{eigenvalues of $H_{\Gamma}^{\omega}\leq E$}\}.
\end{displaymath}
Here we denote the cardinality of a set $S$ by $\#S$.
It is a consequence of ergodic theorem that 
for almost every $\omega$ the limit exists for all $E\in\re$, where, the limit function ${\cal N}$ is continuous, 
and is independent of $\omega$. Moreover $\supp (d{\mathcal N})=\sigma(H^{\omega})$ a.e. $\omega$.
The basic facts about the density of states is found in any of the standard
books in the area for example Cycon-Froese-Kirsch-Simon \cite{CFKS}, 
Carmona-Lacroix \cite{CarLac} and Figotin-Pastur \cite{FigPas}.
It is a result of Pastur \cite{Pastur} and Delyon-Souillard \cite{DS} that ${\cal N}(E)$ is always continuous.
The IDS ${\mathcal N}(E)$ is positive, non-decreasing and bounded (by 1) function
satisfying ${\mathcal N}(\infty) = 1$.
So it is the distribution function of a probability measure.  In the case
when this measure is absolutely continuous, the density $n(E)$ of this measure
is called the   ``the density of states". One of the questions of interest
is the degree of smoothness of the function $n(E)$, which is also often
referred to as the smoothness of IDS, which we do in the following.  

Then our main result is stated as follows:

\begin{thm}
Let $I\Subset I'\Subset \re$ be intervals, and suppose $\m$ has an analytic  density function 
$g(\l)$ on $I'$. Then there is $h_0>0$ such that $n(\l)$ is analytic on $I$ if $0<h<h_0$. 
\end{thm}

Our argument is so simple that we have good control of the constant. 

There are many results on the smoothness of IDS for one-dimensional case. 
For example, ${\cal N}(E)$ is differentiable, even infinitely differentiable under some 
regularity assumptions on $\mu$ (Campanino-Klein \cite{CK} and Simon-Taylor \cite{ST}).
Moreover the smoothness of IDS in the Anderson model on a strip are considered, for example, by 
Klein-Speis \cite{KS1988}, Klein-Lacroix-Speis \cite{KLS1989}, Glaffig \cite{CG} and Klein-Speis \cite{KS1990}.

On the other hand, there are very few results on the smoothness of IDS for 
multi-dimensional case. Using Molchanov formula (of expressing the
matrix elements of $e^{-itH^\omega}$ in terms of a random walk
on the lattice), Carmona showed (see section VI.3 \cite{CarLac} ) that
for the Cauchy distribution the IDS is $C^\infty$.    
Recently, Veseli\'{c} \cite{VESELIC} shows 
the Lipschitz-continuity of IDS for homogeneous Gaussian random potentials using a Wegner estimate. 

Among the most important other results in the multi-dimensional case 
are Bovier-Campanino-Klein-Perez \cite{BCKP},  
Constantinescu-Fr\"ohlich-Spencer\cite{CFS} and Bellissard-Hislop \cite{BelHis}
and all the available
results require that $h$ is small or 
the region of energy considered is away from the middle of the spectrum.
We also consider the case with small $h$, which corresponds to the large disorder case.
A typical result in Bovier-Campanino-Klein-Perez \cite{BCKP} is 
that $\mathcal{N}(E)$ is $(n+1)$-times continuously differentiable 
under the condition that the Fourier transform $\phi(t)$ of $d\mu$ 
satisfies $(1+|t|)^{d+n}\phi(t)\in L^1$.
They also show that if $\phi(t)$ decays exponentially, then $\mathcal{N}(\l)$ has 
an analytic continuation to a strip, provided $h$ is sufficiently small. 
On the other hand Constantinescu-Fr\"ohlich-Spencer \cite{CFS} 
shows that $\mathcal{N}(E)$ is real analytic in $E$, for $|\Re E|$ large enough if
the density of $\mu$ is analytic in the strip $\{V:|\Im V|<2(d+\epsilon)\}$ for 
arbitrarily small, but positive $\epsilon$. 
Bellissard-Hislop \cite{BelHis} proves that if the distribution $d\mu$ 
has a density analytic in a strip about the real axis, then these 
correlation functions are also analytic outside of the planes corresponding to coincident energies. 
In particular, their result implies the analyticity of $n(E)$, 
and of current-current correlation function outside of the diagonal.
The paper \cite{BelHis} employs Taylor expansions to construct analytic continuation, 
whereas we use the Cauchy theorem to show the existence of the analytic continuation. 
This elementary observation greatly simplifies the argument, and also extends the applicability. 
In particular, it makes the argument local in the energy variable, and thus we need the analyticity 
of the density function only locally. 

We prove Theorem~1 in the next section. In Section~3, we discuss an important example, 
i.e., the uniform distribution case, and present explicit constants. Section~4 is devoted to the 
discussion on correlation functions. We consider 2-correlation functions only, 
which are useful 
to study current-current correlation. The idea itself applies
to higher correlation functions.

\section{Proof of Theorem~1}

\subsection{The density of states}

Let $\d_m=(\d_{nm})_{n\in\ze^d}\in\mathcal{H}$ for $m\in\ze^d$, where $\d_{nm}$ is the Kronecker symbol. 
We denote the $(n,m)$-entry of an operator $A$ on $\ell^2(\ze^d)$ 
by $A(n,m)=\jap{\d_n,A\d_m}$.  The following formula of the 
integrated density of states in terms of the spectral projectors is well-known:   
\[
d\mathcal{N}(\l)= \mathbb{E}(E_{H^\o}(d\l,0,0)) = \mathbb{E}\left(\langle \delta_0, E_{H^\o}(d\l)\delta_0 \rangle\right),
\]
where $E_A(\cdot)$ denotes the projection valued spectral measure of 
$A$. 
The
$\mathbb{E}(\cdot)$ denotes the expectation with respect to the randomness
(see, e.g., \cite{CarLac}, Remark~VI.1.5). {Since $\mathbb{E} \left( E_{H^\o}(d\lambda, 0, 0)\right)$ 
is a numerical measure, the general theory of Borel transforms (see Theorem 1.4.16, 
Corollary 1.4.11 of \cite{DemKri}) implies that the following limits exist 
almost every $\lambda$ with respect to the Lebesgue measure
\begin{equation}\label{rep-formula}
n(\l)= \frac1\pi \lim_{\e\to 0} \mathbb{E}(\Im (H^\o-\l-i\e)^{-1}(0,0)).
\end{equation}
and gives the density of the absolutely continuous part of $d\mathcal{N}$. }
Moreover, if $n(\l)$ exists for each point in an interval and bounded, 
then the spectrum is absolutely continuous on the interval by the above formula. 
Hence, in order to prove the analyticity of $n(\l)$, it suffices to show that 
$\mathbb{E}((H^\o-z)^{-1}(0,0))$ is analytic in a complex neighborhood of $I$. 
We use the \textit{random walk expansion}\/ of resolvent to analyze $\mathbb{E}((H^\o-z)^{-1}(0,0))$. 
The random walk expansion is proposed by Fr\"ohlich and Spencer\cite{FroSpe} 
and used by Constantinescu, Fr\"ohlich and Spencer\cite{CFS} for $n(E)$.
\subsection{Random walk expansion of resolvents}

We say $\c=(n_0,n_1,\dots, n_k)\in (\ze^d)^{k+1}$ is a \textit{path of length $k$}\/ if 
$|n_j-n_{j-1}|=1$ for $j=1,\dots, k$, and we write the initial point and the end point of $\c$ 
as $i(\c)=n_0$ and $t(\c)=n_k$, respectively. We write the set of all paths of length $k$ 
with the initial point $n_0$ and the end point $n_k$ by $\Gamma_k(n_0,n_k)$. 
We note $\#\Gamma_k(n_0,n_k)\leq (2d)^k$ for any $n_0,n_k\in\ze^d$. We use the 
Neumann series expansion of the resolvent:
\[
(H^\o-z)^{-1} =(V^\o-z)^{-1} \sum_{k=0}^\infty \bigpare{-H_0 (V^\o-z)^{-1}}^k,
\]
which converges if $|\Im z|> \norm{H_0}=2dh$. It is easy to see, 
{by expanding the summand on the right  hand side and noting that $H_0$ connects only
nearest neighbour sites and that the $(V - z)^{-1}(n, m) = (V(n) - z)^{-1}\delta_{nm}$, } 
\[
\bigpare{H_0 (V^\o-z)^{-1}}^k(n,m) =\sum_{\c\in\Gamma_k(n,m)}
h^k \prod_{j=1}^k (V^\o(n_j)-z)^{-1},
\]
where $\c=(n_0,\dots, n_k)$. Thus we learn 
\begin{equation}\label{eq-1}
\mathbb{E}\bigpare{(H^\o-z)^{-1}(n,m)} 
= \sum_{k=0}^\infty \sum_{\c\in\Gamma_k(n,m)} (-h)^k \mathbb{E}
\biggpare{\prod_{j=0}^k (V^\o(n_j)-z)^{-1}}.
\end{equation}
We now consider $\mathbb{E}\bigpare{\prod_{j=0}^k (V^\o(n_j)-z)^{-1}}$ for each 
$\c=(n_0,\dots, n_k)\in \Gamma_k(n,m)$. We denote 
$\#(\c,\a)=\#\bigset{n_j\in\c}{n_j=\a}\quad \text{for }\a\in\ze^d$.
By the independence of the site potentials, we can write
\begin{align}
\mathbb{E}\biggpare{\prod_{j=0}^k (V^\o(n_j)-z)^{-1}}
&=\prod_{\a\in\ze^d} \mathbb{E}\bigpare{(V^\o(\a)-z)^{-\#(\c,\a)}}\nonumber \\
&= \prod_{\a\in\ze^d} \int \frac{d\m(\l)}{(\l-z)^{\#(\c,\a)}}.\label{eq-2}
\end{align}
We note $\#(\c,\a)=0$ except for finitely many $\a$ and hence the product is a finite 
product. We also note $\sum_{\a\in\ze^d} \#(\c,\a)= k$ for $\c\in\Gamma_k(n,m)$. 

For $\ell\geq 0$, we set 
\[
B_\ell(z)= \int \frac{d\m(\l)}{(\l-z)^\ell},
\]
which is primarily defined for $z\in\co_+$, where $\co_\pm=\{z\in\co\,|\, \pm \Im(z)>0\}$. 

\subsection{Analytic continuation of $B_\ell(z)$}

In the following, we suppose $I=(a,b)$, $I'=(a-\d,b+\d)$ with some $\d>0$, and we 
write $\O_\d=\bigset{z\in\co}{\dist(z,I)<\d}$.
We suppose $\m$ has a density function $g(\l)$ on $I'$, and $g(\l)$ is extended to a 
complex function which is holomorphic in $\O_\d$ and continuous on $\overline{\O_\d}$. 

\begin{lem}
Under the assumptions above, $B_\ell(z)$ is extended to a holomorphic function in 
$\O_\d\cup\co_+$. Moreover, there is $C>1$ such that for any $0<\d'<\d$, 
\[
|B_\ell(z)|\leq C(\d-\d')^{-\ell}, \quad \text{for }z\in \O_{\d'},\; \ell=0,1,\dots. 
\]
\end{lem}

\begin{proof}
Let $\y=\pa\O_\d\cap \overline{\co_-}$ (see Figure 1). 

\begin{figure}[h]
\centering
\includegraphics[width=10cm]{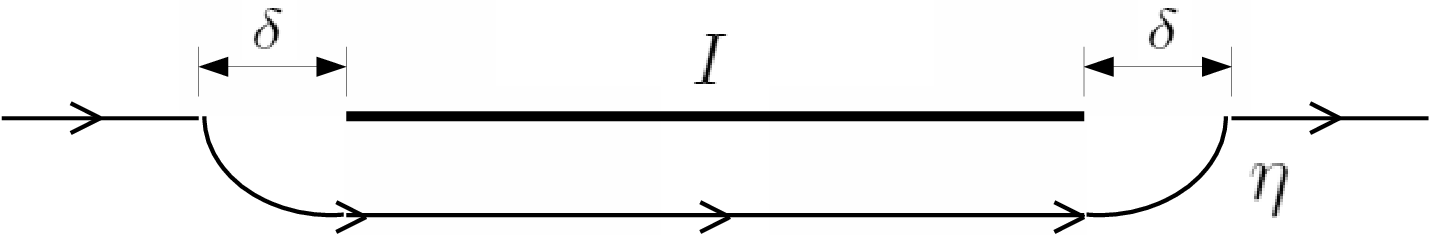}
\caption{The integration path $\y$}
\end{figure}

Then by the Cauchy theorem, we learn 
\[
B_\ell(z)= \int_{\re\setminus I'} \frac{d\m(\l)}{(\l-z)^\ell}
+\int_\y \frac{g(w)dw}{(w-z)^\ell}
\quad \text{for }z\in\co_+.
\]
From this representation, it is clear that $B_\ell(z)$ is extended to $\O_\d\cup\co_+$ 
as a holomorphic function. Also, by setting 
\[
C= 1+(b-a+\pi \d)\sup_{z\in\y}|g(z)|
\]
we have the inequality since $|z-w|\geq \d-\d'$ if $z\in \O_{\d'}$ and $w\in\y$ or 
$w\in\re\setminus I'$. 
\end{proof}

\subsection{Proof of the main theorem}

By Lemma~2 and \eqref{eq-2}, we now learn that $\mathbb{E}\biggpare{\prod_{j=0}^k (V^\o(n_j)-z)^{-1}}$ 
is extended to a holomorphic function in $\co_+\cup \O_\d$, and it is bounded by 
$C^k (\d-\d')^{-k}$ on $\O_{\d'}$. We recall $\#\Gamma_k(n,m)\leq (2d)^k$, and hence 
\begin{align*}
&\sum_{k=0}^\infty \sum_{\c\in\Gamma_k(n,m)} h^k \biggabs{\mathbb{E}
\biggpare{\prod_{j=0}^k (V^\o(n_j)-z)^{-1}}}\\
&\qquad \leq \sum_{k=0}^\infty h^k (2d)^k C^k (\d-\d')^{-k}
= \sum_{k=0}^\infty \biggpare{\frac{2dCh}{\d-\d'}}^k<\infty
\end{align*}
for $z\in \O_{\d'}$ if $h<(\d-\d')/(2dC)$. Thus by \eqref{eq-1}, under this condition, \\
$\mathbb{E}\bigpare{(H^\o-z)^{-1}(n,m)}$ is holomorphic in $z\in\O_{\d'}$. 
In particular, we learn $\mathbb{E}\bigpare{(H^\o-z)^{-1}(0,0)}$ 
is holomorphic in $z\in\O_{\d'}$, and we conclude Theorem~1 thanks to \eqref{rep-formula}. 
\qed

\section{An example}

Here we consider a typical Anderson model, that is, the case when $\m$ has the uniform 
distribution on $[-a,a]$ with $a>0$. In this case, it is well-known $\s(H^\o)=[-a-2dh,a+2dh]$ almost
surely. We also note that it is known that the density of states $n(\l)$ exists and it is 
smooth if $h$ is sufficiently small (\cite{BCKP}, Corollary~1.3). 

Since $\m$ has the density $1/2a$ on $(-a,a)$, 
we can apply Theorem~1 to conclude that for any $b<a$, $n(\l)$ is analytic on $(-b,b)$ 
if $h$ is sufficiently small. However, in this case, we can explicitly compute $B_\ell(z)$ 
to obtain a sharper result. We have 
\[
B_1(z)= \frac{1}{2a}\int_{-a}^a \frac{d\l}{\l-z}
= \frac{1}{2a}(\log(a-z)-\log(-a-z)),
\]
and for $\ell\geq 2$, 
\[
B_\ell(z)= \frac{1}{2a}\int_{-a}^a \frac{d\l}{(\l-z)^\ell}
= \frac{1}{2a(\ell-1)}\biggpare{\frac{1}{(a-z)^{\ell-1}}-\frac{1}{(-a-z)^{\ell-1}}}.
\]
Now we consider the case $h=1$ and change $a$, which is equivalent by a simple scaling. 
. If $\d(\log\d+\pi)\leq a$ and $1\leq \d\leq a$, 
then we have
\[
|B_\ell(z)| \leq \d^{-\ell}
\quad\text{for }z \text{ such that }\dist(z,\{\pm a\})\geq \d, \Re z\in (-a,a), 
\]
for all $\ell\in\mathbb{N}$. 
If $a>2d(\log(2d)+\pi)$, then we can choose $\d>2d$ arbitrarily close to $2d$ to satisfy the above conditions. 
Then by modifying the above argument, we learn that $\mathbb{E}((H^\o-z)^{-1}(0,0))$ 
is analytic on $\{z\,|\,\dist(z,\{\pm a\})>\d,\Re z\in(-a,a)\}$. 
Thus we have the following theorem. 

\begin{thm} Let $h=1$, and 
suppose $\m$ has the uniform distribution on $[-a,a]$ with $a>2d(\log(2d)+\pi)$. 
Then $n(\l)$ is analytic on $(-a+2d, a-2d)$. 
\end{thm}

\section{Correlation functions}
Here we extend our method to show the analyticity of correlation functions (see Bellissard-Hislop \cite{BelHis}). 
For simplicity, we consider 2-correlation functions only, which apply to current-current correlation 
functions. 

We denote the formal expression, where the limit is in the weak operator topology, 
\[
\d(H^\o-E)=\frac{1}{\pi}\lim_{\e\to+0} \Bigbrac{\frac{\e}{(H^\o-E)^2 +\e^2}}
\] 
when the limit is well-defined. Let $A_1$, $A_2$ be bounded operators, which { are } local in 
the following sense: there are $R,M>0$ such that 
\[
A_j(k,\ell)=0 \text{ if }|k-\ell|> R, \quad 
|A_j(k,\ell)|\leq M\text{ for } \forall k,\ell,
\]
with $j=1,2$. The 2-correlation function for $H^\o$, $A_1$ and $A_2$ is defined 
(see the Appendix) by
\[
K(e_1,e_2) =\mathbb{E}\bigbrac{\langle \delta_0, (\d(H^\o-e_1)A_1\d(H^\o-e_2)A_2)\delta_0 \rangle} 
\]
for almost every $e_1,e_2\in\re$. 

\begin{thm}
Let $I$, $I'$ be intervals as in Theorem~1. Then there is $\c>0$ such that 
$K(e_1,e_2)$ is analytic on $I\times I\setminus D(\c h)$, where 
$D(\b)= \bigset{(e_1,e_2)\in \re^2}{|e_1-e_2|\leq \b}$. 
\end{thm}

Namely, the 2-correlation function $K(e_1,e_2)$ is analytic away from a small neighborhood of the 
diagonal, and the width of the exceptional set is $O(h)$ as $h\to 0$. We note that for a fixed 
$h>0$, we do not have the analyticity in $(e_1,e_2)$ very close to the diagonal set. 

\begin{proof}
Let $I=[a,b]$, and let $a\leq E_1<E_2\leq b$. We show 
\[
F(z_1,z_2) =\mathbb{E}\bigbrac{((H^\o-z_1)^{-1}A_1(H^\o-z_2)^{-1}A_2)(0,0)}
\]
is analytically extended to a complex neighborhood of $(z_1,z_2)=(E_1,E_2)$ from 
$(z_1,z_2)\in \co_\pm\times\co_\mp$, or $(z_1,z_2)\in \co_\pm\times\co_\pm$. 
We consider the case:  $(z_1,z_2)\in \co_+\times \co_-$ only. The other cases can be  
handled similarly. We choose $\d>0$ so that
\begin{equation}\label{cond}
E_2\geq E_1+2\d, \quad [a-\d,b+\d]\Subset I'.
\end{equation}
By direct computations as in previous sections, we have 
\begin{align*}
(H^\o-z_1)^{-1}A_1(H^\o-z_2)^{-1}A_2 &= \sum_{k,\ell=0}^\infty 
(V^\o-z_1)^{-1}\bigbrac{(-H_0)(V^\o-z_1)^{-1}}^k \times \\
&\times A_1(V^\o-z_2)^{-1}\bigbrac{(-H_0)(V^\o-z_2)^{-1}}^\ell A_2.
\end{align*}
We denote the set of the paths satisfying the following conditions by $\Gamma_{k,\ell}(n,m)$: 
$\c=(n_0,n_1,\dots, n_k, m_0,\dots, m_\ell, m_{\ell+1})\in (\ze^d)^{k+\ell+3}$ such that 
$n_0=n$, $m_{\ell+1}=m$, $|n_i-n_{i-1}|=1$ for $i=1,\dots, k$, $|m_j-m_{j-1}|=1$ for 
$j=1,\dots, \ell$, $|n_k-m_0|\leq R$ and $|m_\ell-m_{\ell+1}|\leq R$. 
We note $\#\Gamma_{k,\ell}(n,m)\leq (2R)^{2d}(2d)^{k+\ell}$. 
Then we have 
\begin{align*}
&\mathbb{E}\bigbrac{((H^\o-z_1)^{-1}A_1(H^\o-z_2)^{-1}A_2)(n,m)} \\
&\quad = \sum_{k,\ell=0}^\infty \sum_{\c\in \Gamma_{k,\ell}(n,m)} (-h)^{k+\ell} 
A_1(n_k,m_0)A_2(m_\ell,m_{\ell+1})\times \\
&\qquad \qquad \times \mathbb{E}\biggbrac{\prod_{i=1}^k (V^\o(n_i)-z_1)^{-1}
\prod_{j=1}^\ell (V^\o(m_j)-z_2)^{-1}}.
\end{align*}
Now we write 
\[
\n_1(\c,\a)=\#\bigset{n_i}{n_i=\a}, \quad 
\n_2(\c,\a)=\#\bigset{m_j}{m_j=\a}\quad 
\]
for $\c=(n_0,n_1,\dots, n_k, m_0,\dots, m_\ell, m_{\ell+1})\in \Gamma_{k,\ell}(n,m)$.
Then by the independence, it is easy to observe 
\begin{align*}
&\mathbb{E}\biggbrac{\prod_{i=1}^k (V^\o(n_i)-z_1)^{-1}\prod_{j=1}^\ell (V^\o(m_j)-z_2)^{-1}}\\
&\quad=\prod_{\a\in\ze^d} \mathbb{E}\bigbrac{(V^\o(\a)-z_1)^{-\n_1(\c,\a)}
(V^\o(\a)-z_2)^{-\n_2(\c,\a)}}\\
&\quad =\prod_{\a\in\ze^d} \int \frac{d\m(\l)}{(\l-z_1)^{\n_1(\c,\a)}(\l-z_2)^{\n_2(\c,\a)}}.
\end{align*}
We also note $\sum_\a \n_1(\c,\a)=k$ and $\sum_\a \n_2(\c,\a)=\ell$ for $\c\in\Gamma_{k,\ell}(n,m)$. 
We denote 
\[
B_{k,\ell}(z_1,z_2) = \int \frac{d\m(\l)}{(\l-z_1)^{k}(\l-z_2)^{\ell}},
\quad (z_1,z_2)\in\co_+\times \co_-, 
\]
and $\O_\d(E)=\bigset{z\in\co}{|z-E|<\d}$.
Then, as well as Lemma~2, we have the following lemma: 

\begin{lem}
$B_{k,\ell}(z_1,z_2)$ is extended to a holomorphic function in $\O_\d(E_1)\times \O_\d(E_2)$. 
Moreover, there is $C>1$ such that for any $0<\d'<\d$, $k$ and $\ell$, 
\[
|B_{k,\ell}(z_1,z_2)|\leq C (\d-\d')^{-k-\ell}
\quad\text{for } (z_1,z_2)\in \O_{\d}(E_1)\times \O_{\d'}(E_2). 
\]
\end{lem}

We note the constant $C$ is also independent of $E_1$, $E_2$ and $\d,\d'>0$ as long as 
they satisfy \eqref{cond}. The proof of Lemma~5 is similar to that of Lemma~2, but we 
use the contour:
\[
\y=\partial(\co_+\cup \O_\d(E_1)\setminus \O_\d(E_2))
\]
to represent $B_{k,\ell}(z_1,z_2)$ by a contour integral (see Figure 2).
We omit the detail.  
\begin{figure}[h]
\centering
\includegraphics[width=10cm]{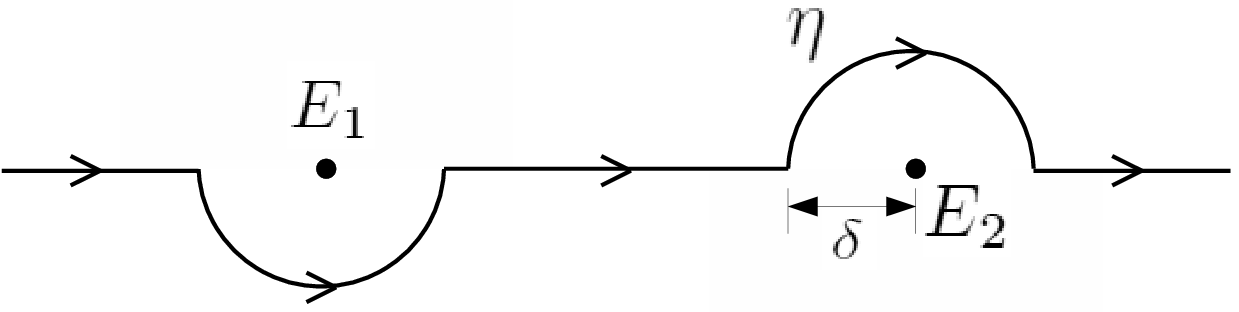}
\caption{The integration path}
\end{figure}

Using the lemma, we learn
\[
|F(z_1,z_2)|\leq \sum_{k,\ell=0}^\infty  (2R)^{2d} (2d)^{k+\ell} h^{k+\ell} M^2 C^{k+\ell}
(\d-\d')^{-k-\ell}
\]
for $(z_1,z_2)\in \O_{\d'}(E_1)\times \O_{\d'}(E_2)$. If 
\[
2d h C(\d-\d')^{-1}<1,\text{ i.e., } \d-\d' >2dCh, 
\]
then the random walk expansion converges uniformly in $\O_{\d'}(E_1)\times \O_{\d'}(E_2)$, 
and in particular, $F(z_1,z_2)$ is analytic. We may choose $\d'=\d/2$, and we conclude 
$F(e_1,e_2)$ is analytic if $|e_1-e_2|>8dCh$. This, combined with other cases, 
implies the assertion. 
\end{proof}

\section{Appendix}

We define the correlation function $K(e_1, e_2)$ here.  Given a pair of vectors
$f, g \in \ell^2({\mathbb Z}^d)$ and any self adjoint operator $A$, 
we see that the limits 
$$
\langle f,  \delta(A-E) g\rangle = lim_{\e \rightarrow 0} \langle f, \frac{\epsilon}{(A-\lambda)^2 + \e^2} g \rangle 
$$
exist for almost every $\lambda$ with respect to the Lebesgue measure, by using Theorem 1.4.16 
\cite{DemKri} and polarization identity to write the finite complex measure 
$\langle f, E_A(\cdot) g\rangle$ as a linear combination of the positive finite measures 
$$
\langle h_\rho , E_A(\cdot) h_\rho \rangle, ~ h_\rho = f+ \rho g ~ \rho \in \{-1,1, 0, i, -i\}.
$$
Thus the function $r(E) = \langle f,  \delta(H - E) g\rangle$  gives the density of 
the measure $\langle f, E_A(\cdot) g\rangle$ almost every $E$.  Similarly computing (and
justifying the exchange of ${\mathbb E}$ and the integrals with respect to $\l_1, \l_2$, by
Fubini,  
\begin{equation}\begin{split}
& {\mathbb E} \langle \delta_0,  \frac{\e}{(H^\o-E)^2 + \e^2 } A_2 \frac{\e}{(H^\o - e_2)^2 + \e^2}A_1 \delta_0 \rangle \\
&  = \int \frac{\e}{(\l_1 - e_1)^2+\e^2} \frac{\e}{(\l_2 - e_2)^2 + \e^2} {\mathbb E} \left( \langle \delta_0, E_{H^\o}(d\l_1) A_2 E_{H^\o}(d\l_2) A_1 \delta_0 \rangle \right) \\
&= \int \frac{\e}{(\l_1 - e_1)^2+\e^2} \frac{\e}{(\l_2 - e_2)^2 + \e^2} \mu(d\l_1, d\l_2), 
\end{split}
\end{equation} 
where $\mu$ is the finite complex correlation measure 
$$
\mu(d\l_1, d\l_2) = {\mathbb E} \left( \langle \delta_0, E_{H^\o}(d\l_1) A_2 E_{H^\o}(d\l_2) A_1 \delta_0 \rangle \right)
$$
As in the proof of 
Theorem 1.4.16 \cite{DemKri}, we can show that the limits 
$$  
K(e_1, e_2) = \lim_{\e \rightarrow 0} \int \frac{\e}{(\l_1 - e_1)^2+\e^2} \frac{\e}{(\l_2 - e_2)^2 + \e^2} \mu(d\l_1, d\l_2)
$$
exist for almost every $(e_1, e_2) \in {\mathbb R}^2$.  
Now note that the complex valued function 
$$
\int \frac{1}{\l_1 - z_1}  \frac{1}{\l_1 - z_2} \mu(d\l_1, d\l_2)
$$
of two complex variables is analytic in a neighborhood of $(e_1,e_2)$ if and only if it is analytic as a function of 
the real variables $(e_1, e_2)$. Therefore in the proof of the
Theorem~4 we consider the function 
$$
F(z_1,z_2) =\mathbb{E}\bigbrac{((H^\o-z_1)^{-1}A_1(H^\o-z_2)^{-1}A_2)(0,0)}
$$ 
and show it is analytic in the region stated in the theorem.


\small
\begin{thebibliography}{99}
%
\bibitem{BelHis} J. Bellissard and P. Hislop:  
Smoothness of Correlations in the Anderson Model at high disorder,
Ann. Henri Poincar\'{e} {\bf 8},  1--26 (2007).
%
\bibitem{BCKP} A. Bovier, M. Campanino, A. Klein and J. F. Perez:  
Smoothness of the Density of States in the Anderson Model at High Disorder, Commun. Math. Phys. {\bf 114}, 439--461 (1988).
%
\bibitem{CK} M. Campanino and A. Klein: A, Supersymmetric Transfer Matrix and Differentiability of the 
Density of States in the One-Dimensional Anderson Model, Commun. Math. Phys. {\bf 104},  227--241 (1986).
%
\bibitem{CarLac} R. Carmona and J. Lacroix:
  \textit{Spectral Theory of Random Schr\"odinger Operators,}
Birkh\"auser, Boston, 1990.
%
\bibitem{CFS} F. Constantinescu, J. Fr\"ohlich, and T. Spencer, Analyticity of the Density of States and 
Replica Method for Random Schr\"odinger Operators on a Lattice, J. Stat. Phys. {\bf 34}, Nos. 3/4,  
571--596 (1984).
%
\bibitem{CFKS} H. Cycon, R. Froese, W. Kirsch and B. Simon:
 \textit {Schr\"odinger Operators},
     {Texts and Monographs in Physics}, Springer-Verlag, 1985.
%
\bibitem{DS} F. Delyon and B. Souillard:
Remark on the Continuity of the Density of States of Ergodic Finite Difference Operators, 
Commun.\ Math.\ Phys.\  \textbf{94}, 289--291(1984).
%
\bibitem{DemKri}M. Demuth and M. Krishna: 
\textit{Determining Spectra in Quantum Theory}, Progress in Mathematical Physics
\textbf{44}, Birkhauser, Boston, 2005.
%
\bibitem{FigPas} A. Figotin and L. Pastur: \textit {Spectra of Random and Almost-Periodic Operators}, 
Springer-Verlag, Berlin, 1992.
%
\bibitem{FroSpe} J. Fr\"ohlich and T. Spencer:  A Rigorous Approach to Anderson Localization, 
in \textit{Common Trends in Particle and Condensed Matter Physics} (Les Houches, 1983), 
Phys.\ Rep.\ {\bf 103},  9--25 (1984).
%
\bibitem{CG} C. Glaffig: Smoothness of the Integrated Density of States on Strips, 
J. Funct. Anal. {\bf 92},  509--534 (1990).
%
\bibitem{KLS1989} A. Klein, J. Lacroix and A. Speis: Regularity of the Density of States 
in the Anderson Model on a Strip for Potentials with Singular Continuous Distributions, 
J. Stat. Phys,  {\bf 57}, Nos.1/2, 65--88 (1989).
%
\bibitem{KS1988} A. Klein and A. Speis: Smoothness of the Density of States in the Anderson Model 
on a One-Dimensional Strip, Ann. Phys. {\bf 183},  352--398 (1988).
%
\bibitem{KS1990} A. Klein and A. Speis: Regularity of the Invariant Measure and of the Density of States 
in the One-Dimensional Anderson Model, J. Funct. Anal. {\bf 88},  211--227 (1990).
%
\bibitem{Pastur} L. Pastur: Spectral Properties of Disordered Systems in One-body Approximation, 
Commun. Math. Phys. {\bf 75}, 179--196  (1980).
%
\bibitem{ST} B. Simon and M. Taylor: Harmonic Analysis on $SL(n,\re)$ and Smoothness of the 
Density of States in the One-Dimensional Anderson Model, Commun. Math. Phys. {\bf 101},  1--19 (1985).
%
\bibitem{VESELIC} I. Veseli\'{c}: Lipschitz-Continuity of the Integrated Density 
of States for Gaussian Random Potentials, Lett. Math. Phys. {\bf 97}, 25--27 (2011).

\end{thebibliography}
\end{document}